\documentclass[12pt, a4paper]{article}
\usepackage[utf8]{inputenc}
\usepackage{dcolumn,lscape}
\usepackage{amsmath,longtable,multicol,dcolumn,tabularx,graphicx,amssymb}
\usepackage{exscale,amsthm,multirow,rotating}
\usepackage{natbib}
\usepackage{bbm}
\usepackage{tikz,dsfont,float}
\usepackage[caption = false]{subfig}
\newcommand{\xvec}{\boldsymbol}
\newcommand{\xmat}{\mathbf}
\newcommand{\xset}{\mathds}

\newtheorem{theorem}{Theorem}

\newtheorem{proposition}{Proposition}

\newcommand{\eps}{\varepsilon}

\DeclareMathOperator*{\argmax}{arg\,max}

\usetikzlibrary{calc}
\usetikzlibrary{positioning}

\begin{document}

\def\spacingset#1{\renewcommand{\baselinestretch}%
{#1}\small\normalsize} \spacingset{1}


\title{\bf Spatial autoregressive fractionally integrated moving average model}
  \author{Philipp Otto\\
  \textit{University of Glasgow, United Kingdom} 
  \and 
  Philipp Sibbertsen\\
  \textit{Leibniz University Hannover, Germany}}
  \maketitle

\begin{abstract}
	In this paper, we introduce the concept of fractional integration for spatial autoregressive models. We show that the range of the dependence can be spatially extended or diminished by introducing a further fractional integration parameter to spatial autoregressive moving average models (SARMA). This new model is called the spatial autoregressive fractionally integrated moving average model, briefly sp-ARFIMA. We show the relation to time-series ARFIMA models and also to (higher-order) spatial autoregressive models. Moreover, an estimation procedure based on the maximum-likelihood principle is introduced and analysed in a series of simulation studies. Eventually, the use of the model is illustrated by an empirical example of atmospheric fine particles, so-called aerosol optical thickness, which is important in weather, climate and environmental science.
\end{abstract}

\noindent%

\spacingset{1.45} 

\section{Introduction}\label{sec:introduction}

Long memory of time series is a well-studied problem in statistics (see, e.g., \citealt{beran2017statistics} for an overview). A process is called to have long memory if the temporal autocorrelation is rather slowly decreasing, e.g. compared to autoregressive processes. For instance, consider a fractional Gaussian noise with $H = d + 0.5$, which coincides with an ARFIMA(0,$d$,0) process
\begin{equation*}
  (1 - B)^d Y_t = \varepsilon_t \, ,
\end{equation*}
where $B$ denotes the backshift operator. This process has temporal long memory. For finite samples $Y_1, \ldots, Y_T$, the model can be rewritten in a vector notation as follows
\begin{equation*}
  (\xmat{I} - \xmat{B})^d \xvec{Y} = \xvec{\varepsilon}
\end{equation*}
with $\xvec{Y} = (Y_t)_{t = 1, \ldots, T}$, $\xvec{\varepsilon} = (\varepsilon_t)_{t = 1, \ldots, T}$, $\xmat{I}$ being the identity matrix and
\begin{equation*}
  \xmat{B} = \left(\begin{array}{cccc}
                       0 & \cdots & 0 & 0     \\
                       1 & \cdots & 0 & 0     \\
                       \vdots & \ddots & \vdots & \vdots      \\
                       0 & \cdots & 1 & 0     \\
                    \end{array} \right) \, .
\end{equation*}
Apparently, the process is a random walk if $d = 1$. Moreover, it is important to note that $\xmat{B}$ is a lower triangular matrix. This ensures that there is some lead-lag relation (i.e., there are future and past values) and that the process is well-defined (i.e., $(\xmat{I} - \xmat{B})$ is non-singular).

Now, consider a spatial setting with $n$ locations $\xvec{s}_1, \ldots, \xvec{s}_n$ instead of time points $1, \ldots, T$. These locations are supposed to lie in a $q$-dimensional space $D \subset \xset{R}^q$. Let $\xvec{Y} = (Y(\xvec{s}_i))_{i = 1, \ldots, N}$. In this case, there is no clear lead-lag relationship between the observations. Thus, the observation at one specific location $\xvec{s}$ influences all adjacent regions, but the adjacent ones usually also influence the observation in $\xvec{s}$. There are no ``future'' and ``past'' observations anymore and, therefore, $\xmat{B}$ is not necessarily a triangular matrix (this would only be the case for directional spatial processes, see, e.g., \citealt{Basak18,merk2021directional}). Thus, further assumptions are needed such that the process is well-defined. However, in general, we define a spatial autoregressive fractionally integrated process analogously by
\begin{equation*}
	(\xmat{I} - \xmat{B})^d \xvec{Y} = \xvec{\varepsilon} \, .
\end{equation*}
In spatial settings, the fractional difference operator $(\xmat{I} - \xmat{B})^d$ serves to control both the spatial autocorrelation and the fractional differencing. In this regard, time-series ARFIMA processes and the spatial autoregressive fractionally integrated are slightly different. Moreover, for $d = 1$, the model
\begin{equation*}
  (\xmat{I} - \xmat{B}) \xvec{Y} = \xvec{\varepsilon}
\end{equation*}
coincides with the commonly known spatial autoregressive model, where $\xmat{B}$ determines the spatial dependence structure. Usually, $\xmat{B}$ is chosen as $\rho \xmat{W}$ with known, prespecified weighting matrices $\tilde{\xmat{W}}_1, \ldots, \tilde{\xmat{W}}_k$ and unknown scalar parameters $\rho_1, \ldots, \rho_k$, which has to be estimated \citep[see, e.g.,][]{Elhorst12}.

In this paper, we extend this important class of models to spatial autoregressive fractionally integrated moving average models (spARFIMA). For this reason, we introduce a parameter $d$, which controls the range and the strength of the spatial dependence. Unlike the spatial autoregressive parameters controlling the degree of spatial dependence on all neighbours, the parameter $d$ influences the shape of the spatial autocorrelation function. That is, this parameter allows to increase the range of the spatial dependence while the process is still stationary. However, we always have to assume that $\xmat{I} - \xmat{B}$ is non-singular, restricting the strength of the spatial dependence and leading to a stationary process. Thus, the interpretation of $d$ differs from the time series case. Nevertheless, such a process can be considered to be long-range dependent in the $q$-dimensional space.

Previous approaches of long-range/memory dependence models for spatial models have mostly focussed on geostatistical settings. In contrast to the spatial econometrics framework, where the spatial dependence is modelled via a suitable spatial weights matrix, which defines the extent of the correlation to all adjacent regions, geostatistical approaches capture the spatial dependence by properly choosing the covariance matrix of a multivariate process. The entries of this covariance matrix usually follow a certain parametric covariance function $C : \xset{R}^{q} \rightarrow \xset{R}^{+}$ depending on the difference between two locations $\xvec{s}_i - \xvec{s}_j$.
In particular,  two-dimensional spatial lattice data has been considered, where the spatial dependence is separable (e.g., \citealt{Robinson06}). That is, the spatial dependence is fully symmetric in both ways for each direction, meaning longitudinal and latitudinal directions. Hence, two separate backward-shift operators can be applied for each index.  They are also called double-geometric processes (cf. \citealt{Leonenko13,Martin79}). \cite{Boissy05} introduce a fractionally integrated spatial model by considering two $d$ parameters, one for each backshift operator. Thus, this process has a symmetric, long-range dependence in each direction and directly extends the long-memory idea in time series analysis to spatial settings (two-dimensional separable and symmetric settings). Further, \cite{Shitan08,Ghodsi09} discussed this model. While \cite{Boissy05,Robinson06} focus on Whittle-type estimations of the long-range parameter, \cite{Beran09} introduced a least-squares estimator. Moreover, a central limit theorem for processes having such kind of spatial dependence has been introduced by \cite{Lahiri16}, applicable even for higher-order and irregular lattices. In contrast to these geostatistical approaches, we focus on so-called spatial econometrics models, which account for spatial autoregressive dependence via weighting matrices.

The remainder of the paper is structured as follows. In the following Section \ref{sec:spARFIMA}, the new spARFIMA process is introduced. We present conditions for the existence and stationarity of such a process, and we also point out the differences between time-series ARFIMA processes and geostatistical long-memory processes that assumed separable spatial correlation. For this new spatial model, a quasi-maximum likelihood estimator is derived in Section \ref{sec:inference}. Furthermore, we carried out an extensive simulation study to show the performance of this QML estimator. The results are presented in Section \ref{sec:MCstudy}. Eventually, the model is applied to a real-world example important in environmental science in Section \ref{sec:example}. More precisely, we analyse raster data on aerosol optical depth with different resolutions. Section \ref{sec:conclusion} concludes the paper.


\section{Spatial autoregressive fractionally integrated model}\label{sec:spARFIMA}

Let $\{Y(\xvec{s}) : \xvec{s} \in D\}$ be a univariate process in the spatial domain $D$. For instance, $D$ could be the two-dimensional space of integers, i.e., $D \subset \mathds{Z}^2$, this would cover classical image processes, such as satellite or microscopic images. In spatial statistics, one would commonly refer to this case as a two-dimensional regular lattice process. In econometrics, however, we are often faced to irregular spatial lattice data, like in the case of polygon data (e.g., county-level data). Thus, we generally assume that $D$ is a subset with a positive volume of the $q$-dimensional real space $\xset{R}^q$. That is, contrary to \cite{Robinson20spatiallongmemory}, we do not restrict ourselves on the case that the process is regularly spaced in two dimensions (i.e., two-dimensional lattice) or that the spatial correlation structure should be symmetric and separable.
%
In our case, the process is observed at a set of $n$ locations, $\{\xvec{s}_1, \ldots, \xvec{s}_n\}$. It is worth noting that this definition also includes spatiotemporal processes if one of the $q$ dimensions is the time axis. For a convenient notation, let $\xvec{Y} = (Y(\xvec{s}_i))_{i = 1,\ldots, n}$ be a random vector of all locations and $\xvec{y} = (y(\xvec{s}_i))_{i = 1,\ldots, n}$ its observation. In spatial econometrics, it is common to assume that the spatial dependence structure is described by a spatial weights matrix $\xmat{B} = (b_{ij})_{i,j = 1, \ldots, n}$. The diagonal elements of $\xmat{B}$ are assumed to be zero to prevent self-influences, i.e., $Y(\xvec{s}_i)$ is influenced by itself. In network modelling, this is also known as self-loops.

We define a spatial autoregressive fractionally integrated  moving average  (spARFIMA) process as follows
\begin{equation}\label{eq:spARFIMA}
  (\xmat{I} - \xmat{B}_1)^d \xvec{Y}  =  \xvec{\alpha} + (\xmat{I} - \xmat{B}_2) \xvec{\varepsilon}\,
\end{equation}
with $\xvec{\varepsilon}$ being a vector of independent and identically distributed random variables. The site-specific intercept $\xvec{\alpha} = (\alpha_1, \ldots, \alpha_n)'$ can also be easily extended to linear regression model $\xmat{X} \xvec{\beta}$. However, we initially focus on the general setting, namely having a site-specific intercept $\xvec{\alpha}$ and general weights matrices $\xmat{B}_1$ and $\xmat{B}_2$ for the autoregressive and moving average term, respectively.

In practice, the intercept is often replaced by a constant vector $\xvec{\alpha} = \alpha \xvec{1}$, and the weighting matrices will be replaced by certain parametric models. In the general case, $\xmat{B}_1$ and $\xmat{B}_2$ would consist of $n(n-1)$ unknown parameters, while there are only $n$ observations. Classical choices of such models are, for instance,
\begin{equation}\label{eq:rhoW}
  \xmat{B}_1 = \rho \xmat{W}_1 \quad \text{and} \quad \xmat{B}_2 = \lambda \xmat{W}_2
\end{equation}
with known, pre-specified matrices $\xmat{W}_1$ and $\xmat{W}_2$, which describe the structure of the spatial dependence, e.g., they could be first-order contiguity, $k$-nearest neighbours, or inverse-distance matrices. Moreover, higher-order dependencies can be modelled by a linear combination
\begin{equation*}
  \xmat{B}_1 = \sum_{i = 1}^{k} \rho_{i} \xmat{W}_{i,1} \,  ,
\end{equation*}
where $\xmat{W}_{i,1}$ is a contiguity matrix having positive weights for neighbours of spatial lag-order $i$ only. The order of the spatial autoregression would be $k$ in this case. However, more commonly, first-order spatial autoregressive models are considered, and higher-order dependencies are directly included in the spatial weighting matrix. Some recent approaches also considered estimating $\xmat{B}$ directly by assuming a certain degree of sparsity (e.g. \citealt{Otto18_lasso,Lam13,Lam16}). Similarly, higher-order spatial lags can be included in the moving average term, but this is only rarely found in practical applications.

The following theorem shows that the process is well-defined under common conditions for spatial autoregressive models. That is, for any positive $d$ there exists a one-to-one mapping between $\xvec{Y}$ and $\xvec{\varepsilon}$, i.e., $\xvec{Y} = \xi^{-1}(\xvec{\varepsilon})$ $\xvec{\varepsilon} = \xi(\xvec{Y})$.
\begin{theorem}\label{theorem:one2onemapping}
  If all diagonal entries of $\xmat{B}_1$ and $\xmat{B}_2$ are zero, $||\xmat{B}_1|| < 1$,  $||\xmat{B}_2|| < 1$, and $d > 0$, the process given by \eqref{eq:spARFIMA} is well-defined and there exists one and only one real-valued sequence $Y(\xvec{s}_1), \ldots, Y(\xvec{s}_n)$ that corresponds to $\varepsilon(\xvec{s}_1), \ldots, \varepsilon(\xvec{s}_n)$. Such a process is called a spatial autoregressive fractionally integrated moving average (spARFIMA) process.
\end{theorem}

\begin{proof}
  The process is well-defined and real-valued if and only if $(\xmat{I} - \xmat{B}_1)^d$ is non-singular.
  Applying a binomial expansion, we get that
    \begin{equation*}
    (\xmat{I} - \xmat{B}_1)^d = \sum_{k = 0}^{\infty} \binom{d}{k} (-1)^k \xmat{B}_1^k \, .
  \end{equation*}
  Because $||\xmat{B}_1|| < 1$, the series $\xmat{B}_1^k$ converges for $k \rightarrow \infty$ and  $(\xmat{I} - \xmat{B}_1)$ is invertible. Then, $\xmat{A} = (\xmat{I} - \xmat{B}_1)^{-1}$ and $\xvec{Y} = \xmat{A}^{d} (\xvec{\alpha} + (\xmat{I} - \xmat{B}_2) \xvec{\varepsilon})$.
  Moreover, if $||\xmat{B}_2|| < 1$ $(\xmat{I} - \xmat{B}_2)$ is non singular as well and there is one-to-one mapping from $\xvec{\varepsilon}$ to $\xvec{Y}$.
\end{proof}

This result makes use of the fact that
\begin{equation}\label{eq:inverse}
  \left((\xmat{I} - \xmat{B})^d\right)^{-1} = \left((\xmat{I} - \xmat{B})^{-1}\right)^{d} \, .
\end{equation}
Thus, there is a close relation to spatial autoregressive models, and many results about the existence of spatial autoregressive models also hold for the fractionally integrated model. To be precise, if the spatial autoregressive process for $d = 1$ is well-defined, also the fractionally integrated version exists. For instance, for the common specification with $\xmat{B} = \rho \xmat{W}$, all results about the range of the unknown parameter $\rho$ are valid. This also includes higher-order models, as demonstrated by \cite{Elhorst12}. However, it is important to note that $d$ should be too large; otherwise, the inverse in \eqref{eq:inverse} gets unreasonably large, and its values are almost identical. From a practical perspective, this means that the process is not causal; that is, the observations cannot be determined by all other observations because the range of the spatial dependence exceeds the spatial domain. Thus, the process tends to have either extremely large or small values. This depends on the spatial setting, i.e., the number of locations, neighbourhood structure, etc.

Like for spatial autoregressive models, we also observe locally varying mean levels and heteroscedastic variances. However, the long-range dependence parameter $d$ only affects the global spill-over effects, i.e., those associated with the autoregressive term. Whereas the moving average term only locally affects the first and second-lag neighbours -- via $\xmat{B}_2$ and $\xmat{B}_2\xmat{B}_2'$ -- the autoregressive has global spill-over effects, which are diminished or strengthened by the parameter $d$. The mean vector and covariance matrix of a spARFIMA process is given in the following proposition.

\begin{proposition}\label{prop:mean_variance}
	Suppose that $\varepsilon$ are identically and independently distributed random errors with mean zero, variance $\sigma_\varepsilon^2$ and finite fourth moments. Moreover, assume that all diagonal entries of $\xmat{B}$ are zero, $||\xmat{B}|| < 1$, and $d > 0$. Then, the spatial autoregressive fractionally integrated process given by \eqref{eq:spARFIMA} has mean
	\begin{equation}
		E(\xvec{Y}) = (\xmat{I} - \xmat{B}_1)^{-d}\xvec{\alpha}
	\end{equation}
	 and covariance matrix
	\begin{equation}
		Cov(\xvec{Y}) = \sigma_\varepsilon^2 (\xmat{I} - \xmat{B}_1)^{-d}(\xmat{I} + \xmat{B}_2 + \xmat{B}_2' + \xmat{B}_2\xmat{B}_2')(\xmat{I} - \xmat{B}_1')^{-d} \, .
	\end{equation}
\end{proposition}

\begin{proof} 
	The process can easily be written in a matrix notation as
	\begin{equation*}
		\xvec{Y} = (\xmat{I} - \xmat{B}_1)^{-d} \left[ \xvec{\alpha} + (\xmat{I} - \xmat{B}_2) \xvec{\varepsilon} \right] \, .
	\end{equation*}
	The result can be obtained by straightforward calculations.
\end{proof}

Below, because of their close similarity, we briefly discuss the relation to higher-order SAR models. Such higher-order models typically include multiple spatially lagged variables. For instance, a second-order spatial autoregressive model results by
\begin{equation}\label{eq:2nd_order_SAR}
	\xvec{Y}  = \xmat{B}_{1,1} \xvec{Y} + \xmat{B}_{1,2} \xvec{Y} + \xvec{\varepsilon} = (\xmat{I} - \xmat{B}_{1,1} - \xmat{B}_{1,2})^{-1} \xvec{\varepsilon}\, .
\end{equation}
In contrast, the polynomial expansion would lead to the following model
\begin{equation}\label{eq:polynomial_order_SAR}
	\xvec{Y}  = (\xmat{I} - \xmat{B}_{1,1}) (\xmat{I} - \xmat{B}_{1,2})\xvec{Y} + \xvec{\varepsilon} \, .
\end{equation}
If $\xmat{B}_1$ and $\xmat{B}_2$ takes the easiest parametric form as defined by \eqref{eq:rhoW}, then the parameter space of $\rho_1$ and $\rho_2$ for the process being stationary is much easier to obtain for \eqref{eq:polynomial_order_SAR} than for \eqref{eq:2nd_order_SAR}, as already pointed out by \cite{Elhorst12}.

\subsection{Illustration of Interaction between $\rho$ and $d$}

Eventually, we illustrate the interaction between the spatial autoregressive dependence implied by the weight matrices and the range parameter $d$ using some numerical examples. For simplicity, we only focus on the spatial autoregressive fractionally integrated process without a moving average component (i.e., $\xmat{B}_2 = \xmat{0}$). In contrast to time-series or directional spatial models, we allowed $\xmat{B}_1$ to be non-triangular. Thus, we have to assume that $\xmat{I} - \xmat{B}_1$ is invertible. This, in turn, limits the overall spatial dependency to a certain extent so that the interpretation of $d$ is different compared to the time-series case. That is, there is certain interaction between $\xmat{B}_1$ and $d$, which we will describe below in more detail for the classic case with $\xmat{B}_1 = \rho \xmat{W}_1$.

In Figures \ref{fig:numerical_examples} and \ref{fig:numerical_examples2}, we have illustrated the influence of the central location $\xvec{s}$ on its neighbours of a $20 \times 20$ spatial lattice for different values of $\rho$ and $d$, where $\xmat{W}_1$ is row-standardised Queen's contiguity matrix in all cases. We particularly focus on processes having a strong spatial autocorrelation, namely $\rho \in \{0.85, 0.9\}$. The dependence of a standard spatial autoregressive model (i.e., $d = 1$) is depicted by the red curves in Figure \ref{fig:numerical_examples}. Obviously, increasing values of $\rho$ (solid vs dashed curves) lead to increased spatial dependence. Since $\xmat{I} - \rho \xmat{W}_1$ must be invertible, the parameter $\rho$ must be smaller than one that again limits the spatial autocorrelation. That is, the fractional integration parameter $d$ allows to increase or diminish the spatial autocorrelation, which can be seen by the blue and black curves for $d = 1.5$ and $d = 0.5$, respectively. The intensity of the spatial dependence implied by the blue curves can only be achieved by choosing $\rho$ very close to one for a spatial autoregressive process. However, such a model is close to the ill-defined case.

Now, one might think that by choosing $\rho$ appropriately, one could also achieve the spatial dependence of any other $d$. However, this is not the case, as we illustrate in Figure \ref{fig:numerical_examples2}. Here, we consider a spatial autoregressive model with $\rho = 0.85$ and computed the distance of the closest models with $d = 1.5$ and $d = 2$. That is, we selected $\rho$ such that the squared distances between the curves is minimised -- this leads to $\rho = 0.702$ and $\rho = 0.588$ for $d = 1.5$ and $d = 2$, respectively. Obviously, these curves differ in a way that a larger value of $d$ increases the size of the spatial dependence to the higher-order neighbours (i.e., the ones with a distance of at least $\sqrt{5} \approx 2.23$), while dependence to directly adjacent regions is decreased. Hence, the shape implied by different values of fractional integration parameter is different compared to the classical autoregressive process. 

\begin{figure}
	\centering
	\includegraphics[width=0.99\textwidth]{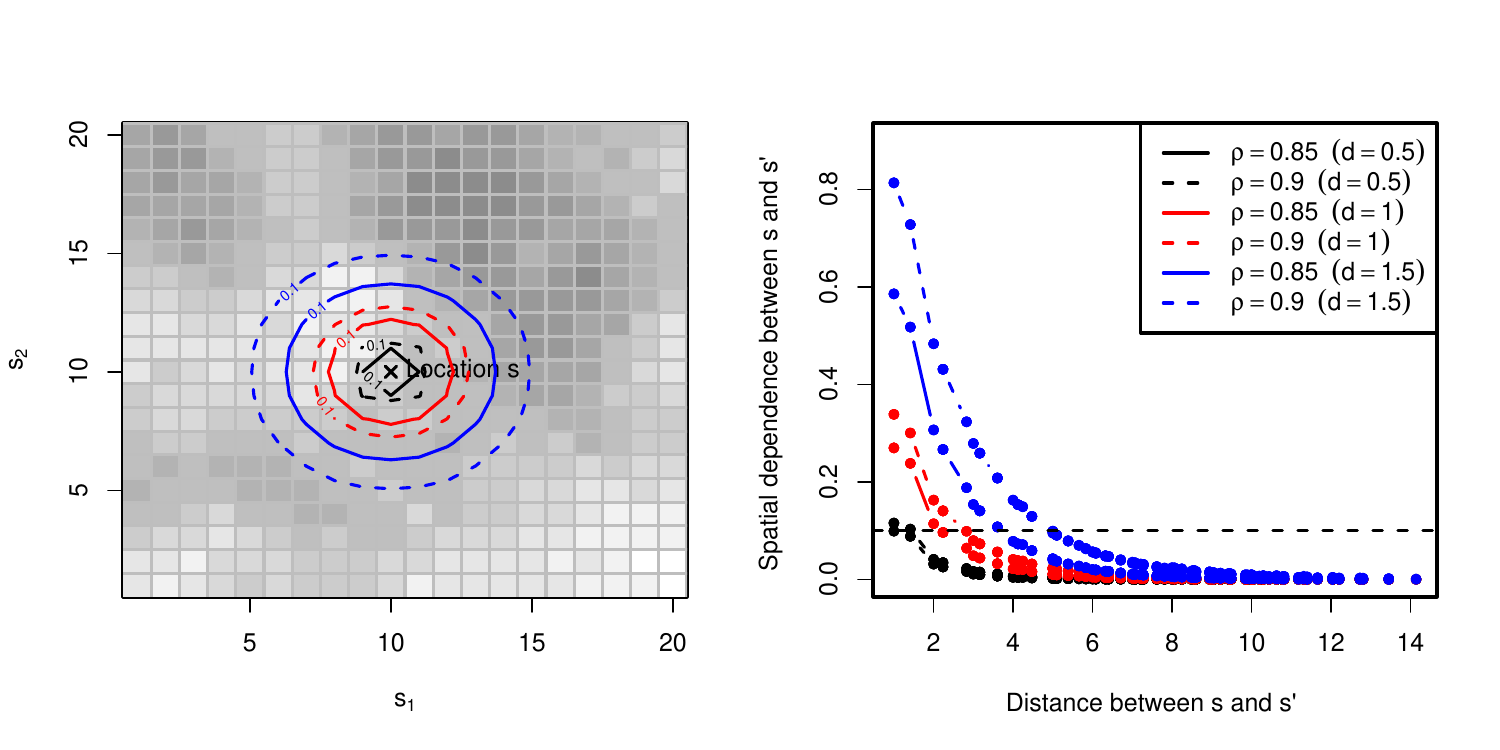}
	\caption{Numerical examples: interaction between $\rho$ and $d$.}\label{fig:numerical_examples}
\end{figure}

\begin{figure}
	\centering
	\includegraphics[width=0.5\textwidth]{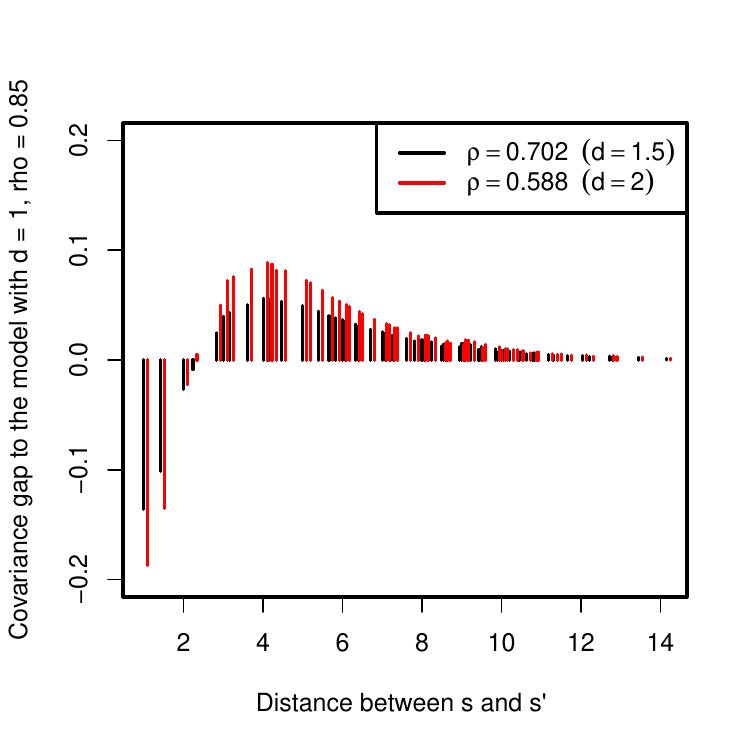}
	\caption{Numerical examples: interaction between $\rho$ and $d$.}\label{fig:numerical_examples2}
\end{figure}

\begin{table}
	\caption{Overview of nested models and special cases.}\label{table:overview}
	\begin{tabular}{p{0.2\textwidth} p{0.22\textwidth} p{0.5\textwidth}} %
		\hline
		Weighting matrix & Spatial dimension $q$ & Resulting model \\
		\hline
		Triangular matrix & 1     & Time-series ARFIMA(0, d, 1) model \\
		                  &       & Note: Adding a further weighting term $(\xmat{I} - \eta \xmat{W})$ with $|\eta| < 1$ leads to an ARIMA(1, d, 1) process \\
		                  & $> 1$ & Causal/directional spatial ARFIMA process \\
		                  &       & Note: With $\xmat{B}_1 = \rho \xmat{W}$ and $\rho = 1$ a non-stationary spatial random walk is obtained if $d = 1$\\
		Non-triangular matrix
		                  & 1     & Non-causal time-series model \\
		                  & $> 1$ & Spatial ARFIMA process \\
		                  &       & Note: $(\xmat{I} - \xmat{B}_1)$ must be invertible (usually, $\xmat{B}_1 = \rho \xmat{W}$ with a known, standardised matrix $\xmat{W}$ and $|\rho| < 1$) to obtain a stationary and well-defined spatial model (i.e., $(\xmat{I} - \xmat{B}_1)^{d}$ serves to control both the fractional differencing and the autoregression)\\
		\hline
	\end{tabular}
\end{table}


\section{QML estimation}\label{sec:inference}

If the spatial dependence structures $\xmat{B}_1$ and $\xmat{B}_2$ are unknown, i.e., it is not known in advance which observations influence each other, each individual link is not generally identifiable. This is a well-known result in spatial econometrics initially pointed out by \cite{Manski93} (see also \citealt{Gibbons12}). For the spatial long-range dependence model, these results hold equivalently. Generally, suppose that there are two different spatial weight matrices $\xmat{B}_1 \neq \xmat{B}_1^*$ and long-range dependence parameters $d \neq d^*$. The model is observationally equivalent if
\begin{eqnarray}\label{eq:uniqueness}
 (\xmat{I} - \xmat{B}_1)^{d} \xvec{u} & = &  (\xmat{I} - \xmat{B}_1^*)^{d} \xvec{u} \, , \, \text{that is,} \\
 (\xmat{I} - \xmat{B}_1)^{d}                 & = &  (\xmat{I} - \xmat{B}_1^*)^{d} \, .
\end{eqnarray}
Here, $\xvec{u}$ denotes the mean and moving average component $\xvec{\alpha} + (\xmat{I} - \xmat{B}_2) \xvec{\varepsilon}$. Thus, if $\xmat{B}_1$ is identifiable, i.e., $\xmat{B}_1 = \xmat{B}_1^*$, and $\xmat{B}_1$ is not equal to a zero matrix, then $d$ is uniquely identifiable. That means that $d$ can only be identified for spatially correlated processes. For the identifiability of $\xmat{B}_1$ all results that hold for spatial autoregressive models can be applied (see, e.g., \citealt{Manski93}). Thus, we follow the common parametric setting described above. That is, suppose that $\xvec{\alpha} = \alpha \xvec{1}$, $\xmat{B}_1 = \rho \xmat{W}_1$, and $\xmat{B}_2 = \lambda \xmat{W}_2$.

Let $\xvec{\varepsilon}$ be a vector of independent and identically distributed random variables with the density $f_{\xvec{\varepsilon}}$. Then, the joint likelihood is given by
\begin{equation}\label{eq:jointprob}
  f_{\xvec{Y}}(\xvec{y}) = \left|(\xmat{I} - \lambda \xmat{W}_2)^{-1}(\xmat{I} - \rho \xmat{W}_1)^d \right| f_{\xvec{\varepsilon}}(\xi(\xvec{y})) \, ,
\end{equation}
where $\xvec{y}$ is the vector of observations. With $f_{\xvec{\varepsilon}}$ being the density of a normal distribution with mean zero and covariance matrix $\sigma^2_\varepsilon \xmat{I}$, the logarithmic likelihood function is obtained as
\begin{equation}\label{eq:ll}
  \mathcal{L}(\xvec{\vartheta} | \xvec{y}) = - \frac{N}{2} \log(2\pi) - \frac{N}{2}\log(\sigma^2_\varepsilon) - \log |\xmat{I} - \lambda\xmat{W}_2| + d \log |\xmat{I} - \rho\xmat{W}_1|  - \frac{1}{2 \sigma^2_\varepsilon} \xi(\xvec{y})'\xi(\xvec{y}) \, .
\end{equation}
The QML estimator of the parameters $\xvec{\vartheta} = (\alpha, \rho, \lambda)'$ is then given by
\begin{equation}\label{eq:qml}
	\hat{\xvec{\vartheta}} = \argmax_{\xvec{\vartheta} \in \Theta} \mathcal{L}(\xvec{\vartheta} | \xvec{y}) \, .
\end{equation}
The parameter space $\Theta$ depends on the choice of the weight matrices $\xmat{W}_1$ and $\xmat{W}_2$, such that the assumptions of Theorem \ref{theorem:one2onemapping} are fulfilled. The main drawback of the QML approach is the scalability to large data sets because it involves the computation of the determinants of the Jacobian, i.e., $ |\xmat{I} - \rho\xmat{W}_1|$ and  $|\xmat{I} - \lambda\xmat{W}_2|$. To avoid repeatedly computing the determinant, we suggest following the approach by \cite{Ord75} for both determinants, that is,
\begin{equation*}
	\log |\xmat{I} - a \xmat{W}| = \sum_{i=1}^{n} \log (1 - a \lambda_W) \, ,
\end{equation*}
where $\lambda_W$ are the eigenvalues of $\xmat{W}$, which have to be computed only once. This is the main bottleneck of the QML approach regarding scalability. An alternative method is the generalised method of moments, for instance (see \cite{dougan2013gmm} for SARMA models).


\section{Simulation Studies}\label{sec:MCstudy}

We conducted various simulation studies to analyse the algorithm's performance and scalability. For all of them, we considered the classical parametric setup defined above ($\xvec{\alpha} = \alpha \xvec{1}$, $\xmat{B}_1 = \rho \xmat{W}_1$, and $\xmat{B}_2 = \lambda \xmat{W}_2$). The $n \times n$ spatial weight matrix $\xmat{W}_1 = \xmat{W}_2 = \xmat{W}$ is a first-order Queen's contiguity matrix, i.e., all surrounding first-lag neighbours are equally affected. This leads to an isotopic setting. Moreover, the locations are assumed to be on a two-dimensional square grid $D = \{\xvec{s} \in \mathds{Z}^2: (0,0)' \leq \xvec{s} \leq (\delta,\delta)'\}$. We simulated the process for increasing dimensions of the field $\delta \in \{15, 20, 25\}$ leading to increasing sample sizes of $n \in \{15^2, 20^2, 25^5\}$.

Firstly, we focus on the fractional integration parameter $d$ and purely autoregressive dependencies. That is, we set $\lambda$ equal to zero. Secondly, we simulated a spARFIMA process with $\lambda = 0.5$. The range parameter was between $0.5$ and $2$, namely $d \in \{0.8, 1, 1.5\}$. For $d = 1$, the classical spatial autoregressive (with/without a moving average term) is obtained, while for $d = 0.5$ the spatial autoregressive effect is diminished, leading to locally constraint spillovers, and for $d > 1$, the range of the spillover effects is increased compared the SAR case. We considered a medium and large spatial autoregressive dependence, namely $\rho = 0.5$ and $\rho = 0.9$. The results of the simulations experiments in terms of the root mean square errors (RMSE) and the average bias of the estimates can be found in Table \ref{tab:1} and \ref{tab:2} for the setting without and with moving average dependencies, respectively. As expected, the MAE decreases with the increasing size of the spatial fields, while the average bias fluctuates around zero for all cases. Moreover, if the magnitude of the spatial dependence is increasing, the estimates of both $\rho$ and $d$ are getting more precise. This can be seen by the decreasing RMSEs.

\begin{sidewaystable}
	\centering
	\caption{Performance of the QML estimators for $(\rho, \sigma^2_\eps, d)'$ and $\lambda = 0$: root mean square errors (RMSE) and average bias in parentheses, 100 replications.}\label{tab:1}
	\begin{tabular}{ c c | c | c | c | c | c | c}
		\hline
		& & \multicolumn{2}{c|}{Estimation of $\rho$} & \multicolumn{2}{c|}{Estimation of $\sigma^2_\eps = 1$} & \multicolumn{2}{c}{Estimation of $d$} \\
		          &            & $\rho = 0.5$   & $\rho = 0.9$   & $\rho = 0.5$   & $\rho = 0.9$   & $\rho = 0.5$  & $\rho = 0.9$ \\
		\hline
		\multicolumn{2}{c|}{$d = 0.8$} & & & & & & \\
		$\quad$   & $n = 15^2$ & 0.310 (0.069)  & 0.143 (-0.040) & 0.098 (-0.018) & 0.099 (-0.004) & 0.831 (0.203) & 0.387 (0.129)  \\
		          & $n = 20^2$ & 0.297 (0.029)  & 0.139 (-0.056) & 0.069 (-0.007) & 0.071 (0.006)  & 0.806 (0.220) & 0.354 (0.105)  \\
		          & $n = 25^2$ & 0.290 (0.038)  & 0.106 (-0.038) & 0.062 (0.000)  & 0.064 (0.011)  & 0.771 (0.269) & 0.235 (0.080)  \\[.1cm]
		\multicolumn{2}{c|}{$d = 1.0$} & & & & & & \\
		          & $n = 15^2$ & 0.280 (0.076)  & 0.119 (-0.027) & 0.098 (-0.020) & 0.099 (-0.005) & 0.732 (0.099) & 0.298 (0.133)  \\
		          & $n = 20^2$ & 0.258 (0.037)  & 0.102 (-0.029) & 0.069 (-0.009) & 0.071 (0.005)  & 0.705 (0.151) & 0.289 (0.104)  \\
		          & $n = 25^2$ & 0.252 (0.048)  & 0.076 (-0.025) & 0.062 (-0.001) & 0.063 (0.010)  & 0.652 (0.131) & 0.198 (0.060)  \\[.1cm]
		\multicolumn{2}{c|}{$d = 1.5$} & & & & & & \\
		          & $n = 15^2$ & 0.216 (0.086)  & 0.058 (-0.012) & 0.100 (-0.024) & 0.098 (-0.007) & 0.560 (-0.121) & 0.238 (0.042)  \\
		          & $n = 20^2$ & 0.188 (0.057)  & 0.055 (-0.018) & 0.069 (-0.012) & 0.070 (0.003)  & 0.534 (-0.044) & 0.212 (0.065)  \\
		          & $n = 25^2$ & 0.186 (0.062)  & 0.045 (-0.012) & 0.063 (-0.005) & 0.063 (0.008)  & 0.519 (-0.069) & 0.170 (0.040)  \\[.1cm]
		\hline
	\end{tabular}
\end{sidewaystable}

\begin{sidewaystable}
	\caption{Performance of the QML estimators for $(\rho, \sigma^2_\eps, d)'$ and $\lambda = 0.5$: root mean square errors (RMSE) and average bias in parentheses, 100 replications.}\label{tab:2}
	\begin{tabular}{ c c | c | c | c | c | c | c | c | c}
		\hline
		& & \multicolumn{2}{c|}{Estimation of $\rho$} & \multicolumn{2}{c|}{Estimation of $\sigma^2_\eps = 1$} & \multicolumn{2}{c|}{Estimation of $d$} & \multicolumn{2}{c}{Estimation of $\lambda = 0.5$} \\
		          &            & $\rho = 0.5$   & $\rho = 0.9$   & $\rho = 0.5$   & $\rho = 0.9$   & $\rho = 0.5$  & $\rho = 0.9$ & $\rho = 0.5$  & $\rho = 0.9$ \\
		\hline
		\multicolumn{2}{c|}{$d = 0.8$} & & & & & & & & \\
		$\qquad$   & $n = 15^2$ & 0.351 (0.126)   & 0.166 (-0.03)  & 0.333 (0.083)  & 0.287 (-0.064) & 0.099 (-0.011) & 0.099 (0.001)  & 0.663 (0.042) & 0.325 (0.033) \\
		          & $n = 20^2$ & 0.346 (0.187)   & 0.155 (-0.03)  & 0.305 (0.104)  & 0.274 (-0.051) & 0.070 (0.001) & 0.071 (0.011)  & 0.647 (0.119) & 0.308 (0.051) \\
		          & $n = 25^2$ & 0.322 (0.151)   & 0.082 (-0.00)  & 0.288 (0.082)  & 0.255 (-0.044) & 0.064 (0.006) & 0.065 (0.015)  & 0.615 (0.038) & 0.271 (-0.010) \\[.1cm]
		\multicolumn{2}{c|}{$d = 1.0$} & & & & & & & & \\
		          & $n = 15^2$ & 0.384 (0.142)   & 0.104 (-0.010) & 0.350 (0.002)  & 0.320 (-0.137) & 0.099 (-0.016) & 0.098 (-0.001) & 0.643 (-0.036) & 0.324 (-0.115) \\
		          & $n = 20^2$ & 0.365 (0.124)   & 0.101 (-0.008) & 0.328 (0.018)  & 0.307 (-0.172) & 0.070 (-0.003) & 0.071 (0.011) & 0.637 (0.028) & 0.311 (-0.122) \\
		          & $n = 25^2$ & 0.355 (0.133)   & 0.065 (0.003) & 0.340 (0.023)  & 0.262 (-0.143) & 0.063 (0.001) & 0.064 (0.014) & 0.608 (0.004) & 0.284 (-0.124) \\[.1cm]
		\multicolumn{2}{c|}{$d = 1.5$} & & & & & & & & \\
		          & $n = 15^2$ & 0.321 (0.017)  & 0.063 (0.002)  & 0.395 (-0.210) & 0.321 (-0.168) & 0.101 (-0.023) & 0.098 (-0.001) & 0.610 (-0.246) & 0.344 (-0.149) \\
		          & $n = 20^2$ & 0.320 (-0.032)  & 0.050 (0.004)  & 0.393 (-0.253) & 0.312 (-0.195) & 0.069 (-0.012) & 0.070 (0.009) & 0.597 (-0.172) & 0.324 (-0.176) \\
		          & $n = 25^2$ & 0.313 (-0.026)  & 0.040 (0.008)  & 0.387 (-0.243) & 0.286 (-0.174) & 0.062 (-0.005) & 0.064 (0.014) & 0.595 (-0.198) & 0.312 (-0.174) \\[.1cm]
		\hline
	\end{tabular}
\end{sidewaystable}

We also computed the average time needed to estimate the parameters using a standard \texttt{R} implementation for all simulations. The eigenvalues of the weight matrix were computed using the \texttt{eigen} function in \texttt{R}, and the optimisation of \eqref{eq:qml} was done numerically using the algorithm implemented in \texttt{solnp()} (see \citealt{RSolnp}). The computation time is shown in Figure \ref{fig:computation_time} for both simulation studies, i.e., with $\lambda = 0$ and $\lambda = 0.5$.

\begin{figure}
  \centering
  \includegraphics[width=0.49\textwidth]{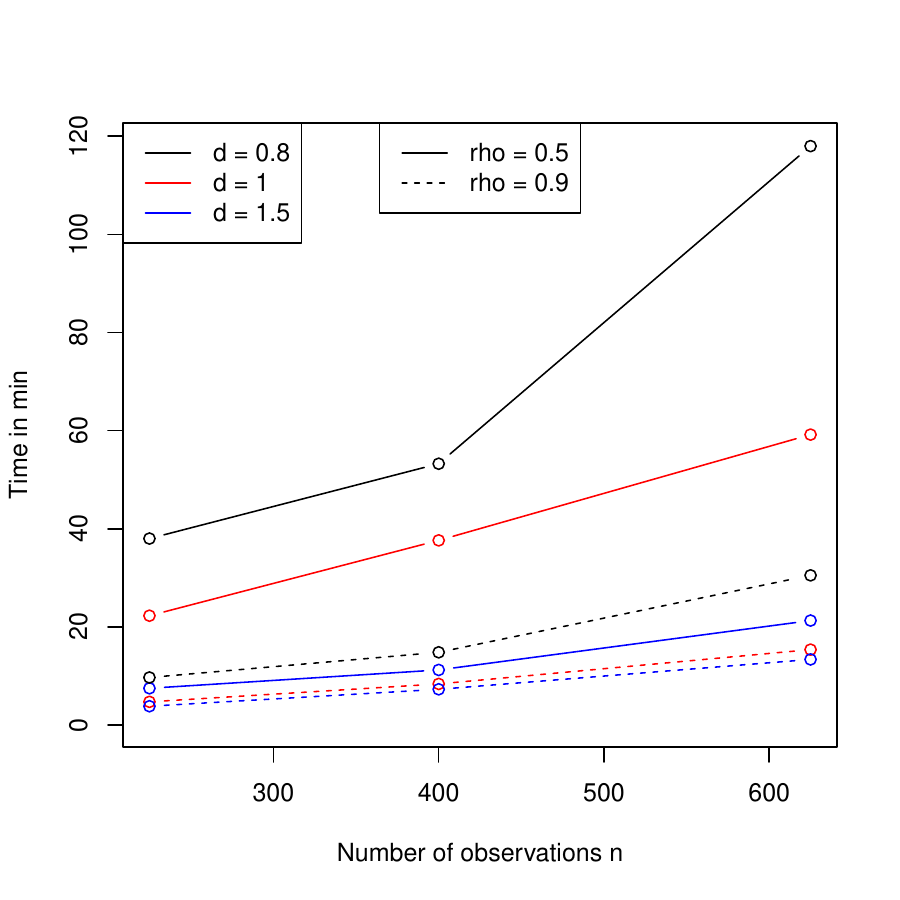}
  \includegraphics[width=0.49\textwidth]{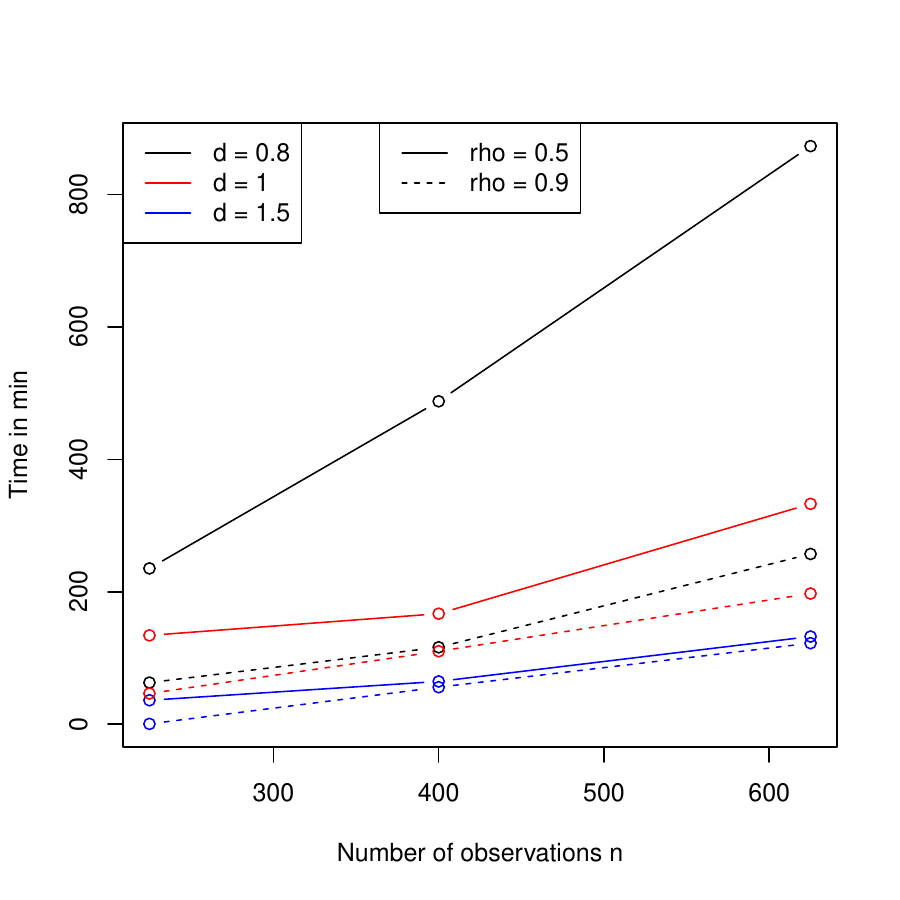}
  \caption{Average computation time for the estimation of the parameters (numerical maximisation) for the case autoregressive model, $\lambda = 0$ (left), and the autoregressive moving average model, $\lambda = 0.5$ (right)}\label{fig:computation_time}
\end{figure}

%
%

%


\section{Real-world illustrative example}\label{sec:example}

To illustrate the usefulness of the range parameter $d$ in practice, we will examine a real-world example below. For this reason, we consider a specific set-up, namely the identical data set, but in three different resolutions. At the same time, we apply classical weighting matrices (i.e., Queen's contiguity matrices) that weigh all neighbouring grid cells equally. Since the data set does not change and the dependence structure (in a geographical sense) thus remains the same, the shape of the spatial autocorrelation function changes for higher resolutions (because the neighbouring raster cells are geographically closer). In the lowest resolution, the distance between grid cells is greater in a geographical sense, and spatial dependence is, therefore, faster, declining to zero (in terms of the number of spatial lags). Hence, the parameter $d$ provides additional flexibility for the shape and the range of the spatial dependence.

More precisely, we consider raster data on the aerosol optical depth obtained from NEO, NASA Earth Observations, measured by NASA's Moderate Resolution Imaging Spectroradiometer (MODIS). The aerosol optical thickness measures the concentration of solid and liquid particles in the atmosphere, so-called aerosols. This aerosol concentration plays an important role in weather, climate, air quality, and thus human's health (cf. \citealt{kumar2007empirical,wang2003intercomparison,gupta2013modis,van2010global}). Moreover, these aerosols are one of the greatest sources of uncertainty in climate modelling.

A climactic active and interesting area is the Northern Atlantic Ocean over the equator. Thus, we considered this area, N $0^\circ$-$25^\circ$, E $-45^\circ$-$20^\circ$, which also covers the area of the most Northern Atlantic hurricanes, in different resolutions of $0.5 \times 0.5$,  $1 \times 1$, and  $2 \times 2$ degrees. This leads to quadratic lattices of sizes $12 \times 12$, $25 \times 25$, and $50 \times 50$ for the highest, medium, and lowest resolution, respectively. Hence, the sample size increases from $n = 144$, $n = 625$, and $n = 2,500$ observations. It is worth noting that this implies a $2,500$-dimensional weighting matrix for the computationally largest problem. Because the focus is on the range of the spatial dependence, we standardised each data set in advance.

The full data set is shown along with the subset of the three considered resolutions in Figure \ref{fig:empiricalexample}. In addition to showing the data, we also provide the estimated spatial autocorrelation functions based on Moran's $I$ in the bottom row of this figure. From these spatial ACFs, one could see that the lower the resolution, the faster the spatial autocorrelation is decaying -- because the directly neighbouring pixels for the lowest resolution already cover a larger geographical distance than the directly adjacent pixels for the highest resolution. Thus, the fractional integration parameter $d$ can provide further flexibility for the model, especially for the larger ranges in higher resolutions. Moreover, one could see that the clusters appear more pronounced with rather sharp edges for the images with a higher resolution compared to the third case with a low resolution.
In Table \ref{tab:example}, we report the resulting estimated parameter along with their estimated standard errors of a spatial ARFIMA model for all three resolutions. The standard errors are obtained from the Hessian of log-likelihood as Cramer-Rao bounds. Because the moving average component seems to be irrelevant (non-significant and leading to lower AIC/BIC), all models have been estimated for $\lambda = 0$. As a benchmark model, we also report the results of a classical spatial autoregressive model (i.e., $d = 1$). For the sake of completeness, we also report the results of a SARMA model. At this point, it is worth noting that one could also test for the difference of the parameter $d$ to 1.

Looking at the information criteria reported in Table \ref{tab:example}, we see that the fractional integration of the spatial autoregressive is particularly useful for medium and high resolutions. While we are getting good model fits for a SAR process in the case of the lowest resolution, both the AIC and BIC criteria are smaller for the spARFIMA process in the two other cases. Moreover, we see that the autoregressive parameters are larger while the parameter $d$ is smaller compared to the low-resolution case. That is, there is a strong spatial dependence on the directly adjacent pixels, which decays fast with the spatial distance. This leads to more pronounced and sharp clusters compared to the low-resolution case, where the clusters rather fade out across space (because of the averaging of the grid cells). To a limited degree, the moving average residuals could also capture this behaviour. Thus, the SARMA model shows a better fit compared to the SAR model for medium and high resolution.

\begin{figure}
	\centering
	\includegraphics[width=0.8\textwidth]{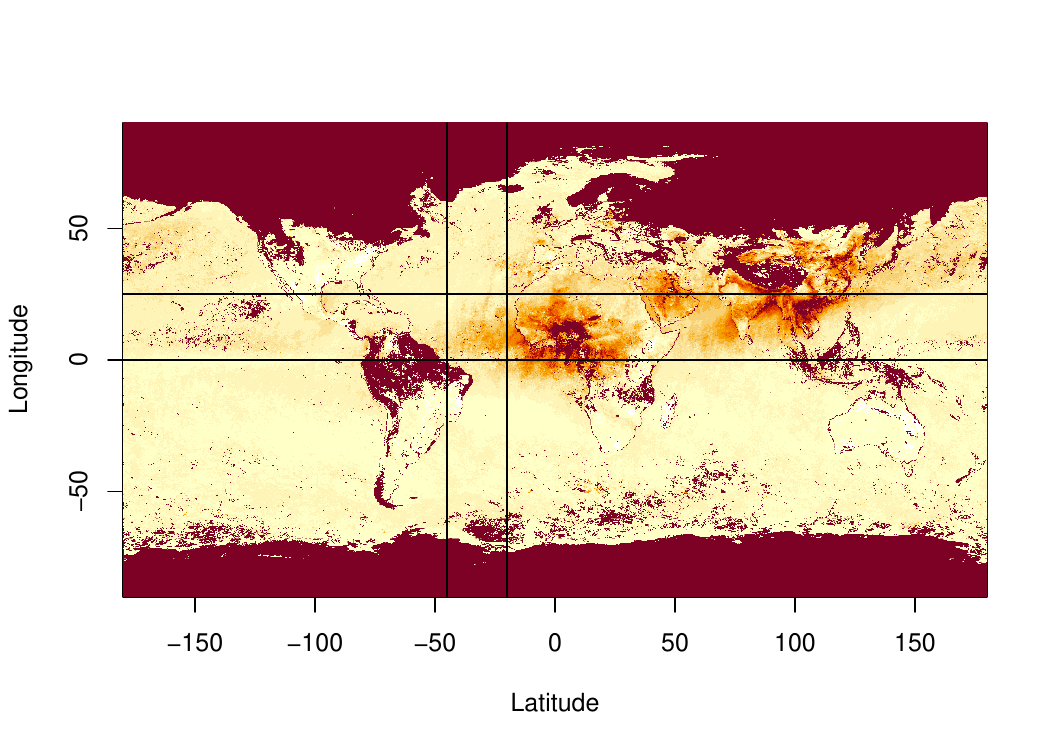}\\
	\includegraphics[width=0.32\textwidth]{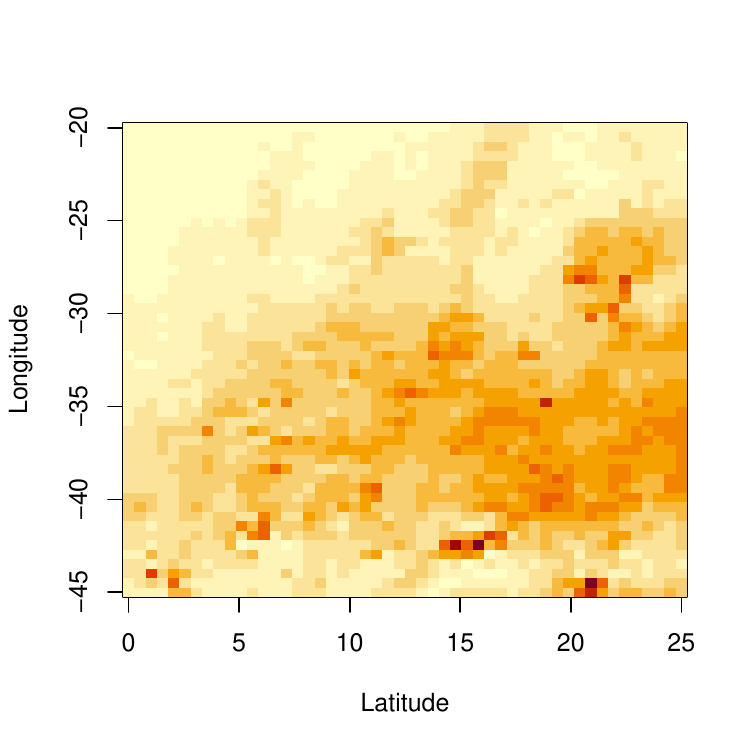}
	\includegraphics[width=0.32\textwidth]{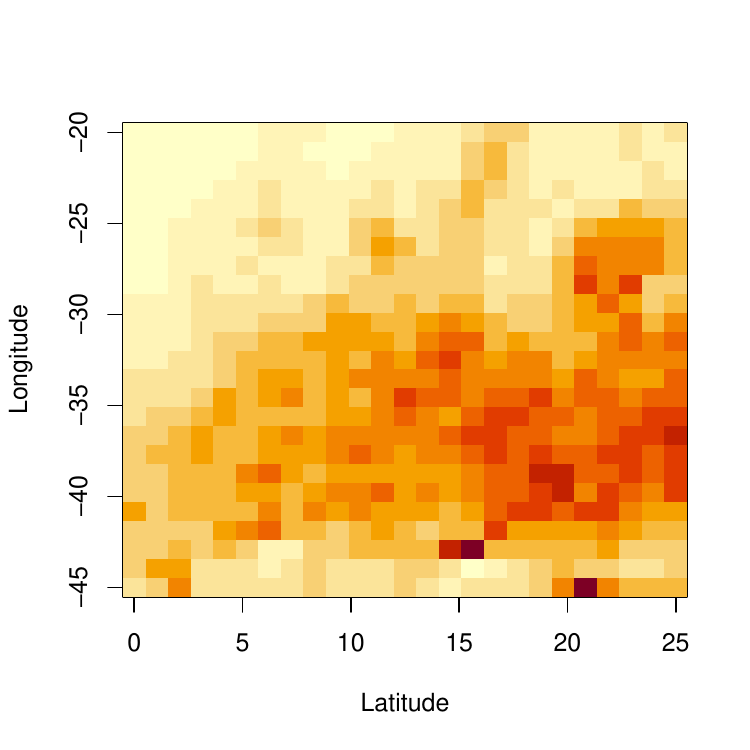}
	\includegraphics[width=0.32\textwidth]{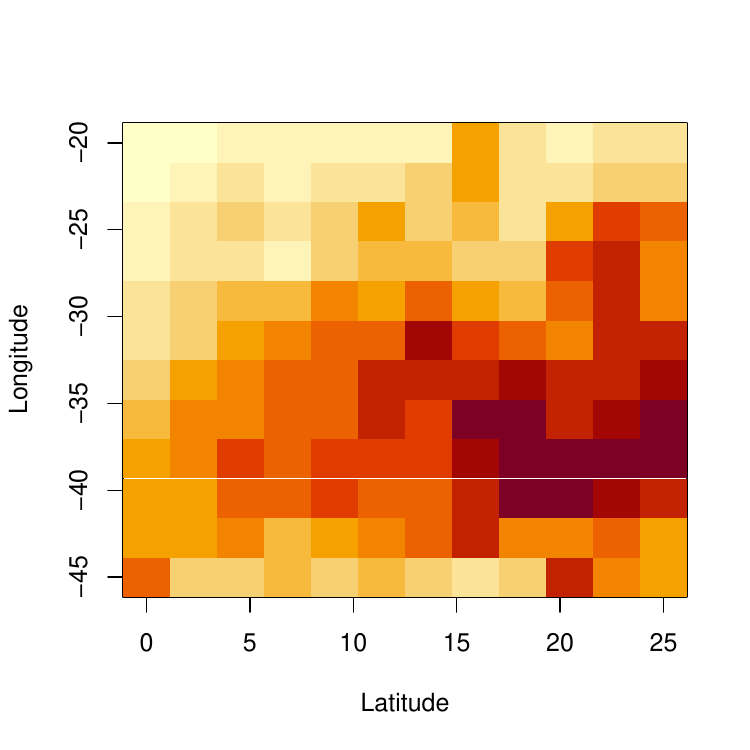}\\
	\includegraphics[width=0.32\textwidth]{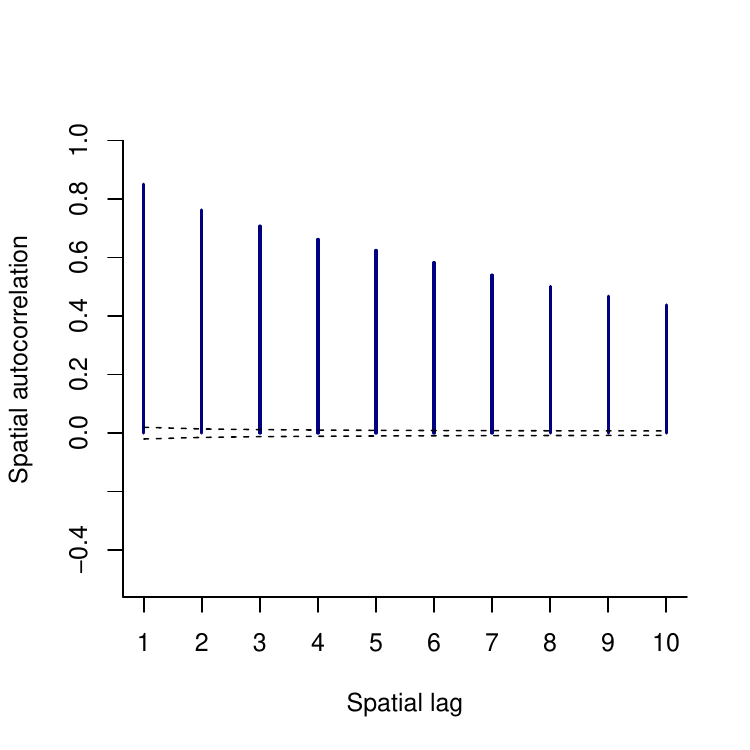}
	\includegraphics[width=0.32\textwidth]{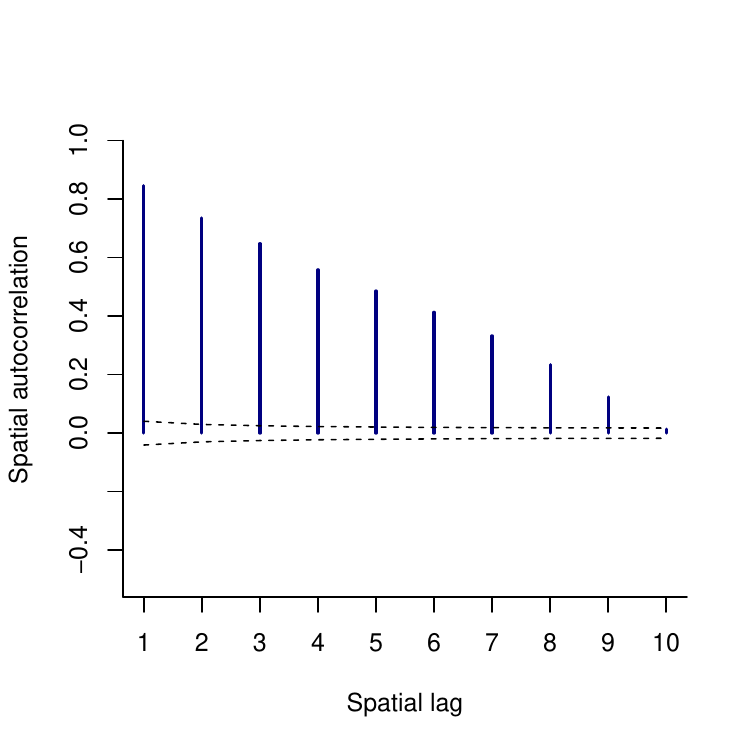}
	\includegraphics[width=0.32\textwidth]{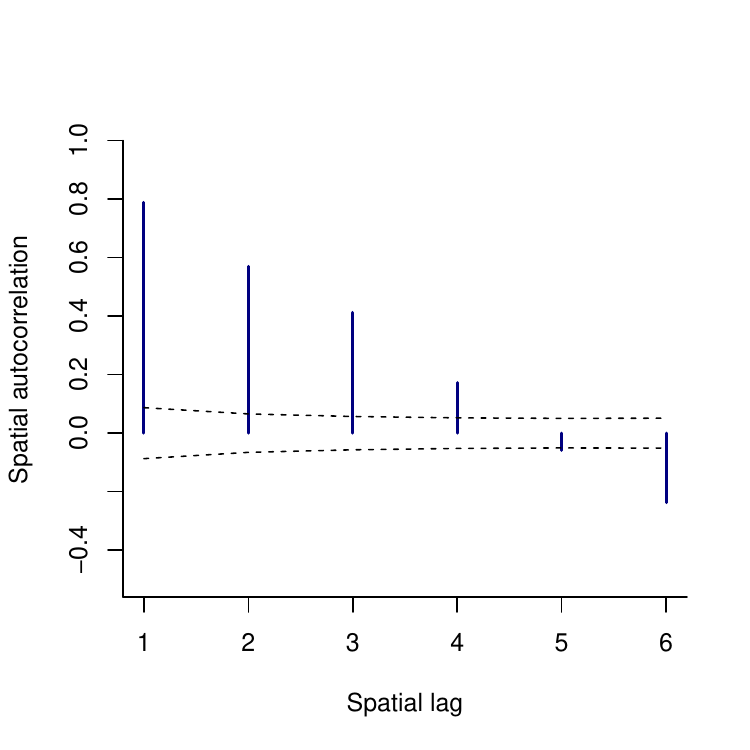}\\
	\caption{Optical aerosol depth (top: global data, middle: high, medium, and low resolution from left to right, bottom: spatial ACF)}\label{fig:empiricalexample}
\end{figure}

\begin{table}
	\centering
	\caption{Estimated parameters of an spatial ARFIMA process (with $\lambda = 0$) and classical SAR and SARMA models as benchmark.}\label{tab:example}
	\begin{scriptsize}
		\begin{tabular}{cc cc cc cc}
			\hline
			     &                    & \multicolumn{2}{c}{spARFIMA} & \multicolumn{2}{c}{SAR model} & \multicolumn{2}{c}{SARMA model} \\
			     & Resolution         & Estimate & Standard error    & Estimate & Standard error     & Estimate & Standard error        \\
			\hline
			                          & low    & 1.3178 & 0.4029 &     &                        &     &                        \\
			$d$                       & medium & 0.7656 & 0.0610 &  \multicolumn{2}{c}{($d=1$)} &  \multicolumn{2}{c}{($d=1$)} \\
			                          & high   & 0.6927 & 0.0271 &     &                        &     &                        \\
			\hline
			                          & low    & 0.8576 & 0.1228 & 0.9440  & 0.0275 & 0.9440  & 0.0373 \\
			$\rho$                    & medium & 0.9911 & 0.0090 & 0.9466  & 0.0126 & 0.9719  & 0.0119 \\
			                          & high   & 0.9967 & 0.0025 & 0.9435  & 0.0066 & 0.9758  & 0.0054 \\
			\hline
			                          & low    &     &                             &      &                            & 0.0000 & 0.2102 \\
			$\lambda$                 & medium & \multicolumn{2}{c}{($\lambda=0$)} & \multicolumn{2}{c}{($\lambda=0$)} & 0.2526 & 0.0913 \\
									  & high   &     &                             &      &                            & 0.3275 & 0.0435 \\
			\hline
			                          & low    & 0.1654 & 0.0216 & 0.1637  & 0.0200 & 0.1637  & 0.0203 \\
			$\sigma_\varepsilon^2$    & medium & 0.1338 & 0.0077 & 0.1347  & 0.0085 & 0.1306  & 0.0076 \\
			                          & high   & 0.1370 & 0.0040 & 0.1378  & 0.0040 & 0.1328  & 0.0039 \\
			\hline
					                  & low    & \multicolumn{2}{c}{189.1391}     & \multicolumn{2}{c}{188.3013*} & \multicolumn{2}{c}{190.3013} \\
			AIC                       & medium & \multicolumn{2}{c}{660.91*}      & \multicolumn{2}{c}{667.1188}  & \multicolumn{2}{c}{662.5683} \\
			                          & high   & \multicolumn{2}{c}{2620.552*}    & \multicolumn{2}{c}{2677.81}   & \multicolumn{2}{c}{2634.109} \\
		    \hline
			                          & low    & \multicolumn{2}{c}{198.0485}     & \multicolumn{2}{c}{194.2409*} & \multicolumn{2}{c}{199.2107} \\
			BIC                       & medium & \multicolumn{2}{c}{674.2232*}    & \multicolumn{2}{c}{675.9943}  & \multicolumn{2}{c}{675.8815} \\
			                          & high   & \multicolumn{2}{c}{2638.024*}    & \multicolumn{2}{c}{2689.462}  & \multicolumn{2}{c}{2651.581} \\
			\hline
			Residuals'  	          & low    & \multicolumn{2}{c}{0.4076}       & \multicolumn{2}{c}{0.4053}    & \multicolumn{2}{c}{0.4053} \\
			standard                  & medium & \multicolumn{2}{c}{0.3661}       & \multicolumn{2}{c}{0.3673}    & \multicolumn{2}{c}{0.3617}\\
			deviation                 & high   & \multicolumn{2}{c}{0.3702}       & \multicolumn{2}{c}{0.3713}    & \multicolumn{2}{c}{0.3644} \\
			\hline
			Moran's $I$	of            & low    & \multicolumn{2}{c}{0.0003 (0.4346)}    & \multicolumn{2}{c}{ 0.0331 (0.1822)}   & \multicolumn{2}{c}{0.0331 (0.1822)}  \\
	        the residuals             & medium & \multicolumn{2}{c}{0.0025 (0.4207)}    & \multicolumn{2}{c}{-0.0386 (0.9645)}   & \multicolumn{2}{c}{0.0006 (0.4567)}  \\
	        (p-value)                 & high   & \multicolumn{2}{c}{0.0053 (0.2867)}    & \multicolumn{2}{c}{-0.0509 (1.0000)}   & \multicolumn{2}{c}{0.0008 (0.4535)}  \\
	        \hline
		\end{tabular}
	\end{scriptsize}
\end{table}


\section{Discussion and Conclusions}\label{sec:conclusion}

Motivated by time-series fractionally integrated autoregressive models, we have introduced the concept of fractional integration for spatial autoregressive processes. More precisely, we developed a spatial autoregressive fractionally integrated moving average model (spatial ARFIMA) that is suitable for data observed in multidimensional space. Moreover, we do not restrict the process to regularly spaced grid data so that the process can be applied to irregular polygon data, as it is often the case in economics, but also to regular grids, like image, geostatistical, or raster data. The latter examples are often present in environmental studies.

In contrast to time-series ARFIMA processes, fractional integration is directly included in the spatial autoregressive term. Alternatively, two different spatial weight matrices could be considered -- one for the fractional integration and one for the autoregressive dependence. In spatial settings, however, the choice of the weight matrix is complicated, and often it has a prespecified structure, so it is preferable to combine these two effects into one term.

This new spatial ARFIMA model is closely related to SAR models, so many results can be directly applied, e.g., on the identification or estimation. This paper considers the frequently applied QML approach to estimate the parameters. We paid particular attention to the scalability of this approach. Furthermore, we analysed the performance and the computation time in a series of Monte-Carlo simulation studies.

Finally, the model has been applied to real data -- aerosol optical depth. We focussed on the interaction between the fractional integration and the spatial autoregressive parameter because the same data was analysed in different resolutions. We found a pronounced spatial dependence for all resolutions. The fractional integration parameter was particularly useful for images in higher resolutions.




\end{document}